\documentclass[%
 reprint,
 amsmath,amssymb,
 aps,
pra,
]{revtex4-2}

\usepackage{graphicx}
\usepackage{dcolumn}
\usepackage{amsmath}
\usepackage{amsthm}
\usepackage{enumitem}
\usepackage{bm}
\usepackage{hyperref}
\usepackage{physics}
\usepackage{xcolor}
\usepackage{graphicx}

\newtheorem{theorem}{Theorem}
\newtheorem{lemma}{Lemma}
\newtheorem{definition}{Definition}

\newcommand{\cE}{\mathcal{E}}

\newcommand{\cH}{\mathcal{H}}
\newcommand{\cI}{\mathcal{I}}

\newcommand{\cU}{\mathcal{U}}

\begin{document}

\preprint{APS/123-QED}

\title{A short note on the 0-fidelity}

\author{Karl Mayer}

\affiliation{Honeywell Quantum Solutions}
\date{\today}

\begin{abstract}
A recent article introduced a hierarchy of quantities 
called $k$-fidelities that approximate the quantum process fidelity with increasing accuracy.
The lowest approximation in this hierarchy is the $0$-fidelity~\cite{greenaway2021}.
The authors gave a protocol for estimating the 0-fidelity
and showed numerical evidence that it approximates the process fidelity.
In this note, we prove lower and upper bounds on the process fidelity as linear functions of the 0-fidelity.
The proof is a simple modification of a proof in~\cite{Mayer2018}.
By solving a semidefinite program,
we provide evidence that the lower bound is tight.

\end{abstract}

\maketitle

\section{Introduction}\label{sec:level1}

Let $\cE$ be a quantum process, given by a completely-positive trace preserving (CPTP) map,
acting on the state space of $n$ qubits.
The 0-fidelity of $\cE$ was introduced in~\cite{greenaway2021}
and is defined in terms
of its action on tensor products of single-qubit symmetric positive operator-valued measure (SIC-POVM) states.
We recall that a set of 4 pure qubit states $\{\ket{\psi_i}\}_{i=1}^4$
is a called a single-qubit SIC-POVM if
\begin{equation}
    |\braket{\psi_i}{\psi_j}|^2=\frac{1}{3}\quad i\ne j .
\end{equation}
For an $n$-qubit system, with dimension $d=2^n$, consider the $d^2$ states
\begin{equation}
    \ket{\psi_{k_1\cdots k_n}} = \ket{\psi_{k_1}}\otimes\cdots\otimes\ket{\psi_{k_n}},
\end{equation}
for $k_i\in\{1,2,3,4\}$, and where for all $i$, the set of states
$\{\ket{\psi_{k_i}}\}_{k_i}$ is a single-qubit SIC-POVM.
We index these states with a multi-index $k=k_1\cdots k_n$.
\begin{definition}
    The $0$-fidelity of $\cE$, with respect to the
    identity, is given by
    \begin{equation}\label{eq: 0-fidelity}
        F_0(\cE)=\frac{1}{d^2}\sum_{k=1}^{d^2}\bra{\psi_k}\cE\big(\ket{\psi_k}\bra{\psi_k}\big)\ket{\psi_k},
    \end{equation}
    for states $\ket{\psi_k}=\ket{\psi_{k_i}}\otimes\cdots\otimes\ket{\psi_{k_n}}$ such that for all $i$, $\{\ket{\psi_{k_i}}\}_{k_i=1}^4$ is a single-qubit SIC-POVM.
\end{definition}

The $0$-fidelity does not depend on the choice of
single-qubit SIC-POVMs. Put another way, we have
$F_0(\cE)=\cU\circ\cE\circ\cU^{\dagger}$ for any
local unitary $\cU=\cU_1\otimes\cdots\otimes\,\cU_n$.
This follows from the facts that $F_0(\cE)$ is quadratic in all $\ket{\psi_{k_i}}\bra{\psi_{k_i}}$ and that SIC-POVMs are 2-designs~\cite{Renes2004}.

The process fidelity, also known as the entanglement fidelity~\cite{Nielsen2002}, is defined as follows.
Let $\{\ket{x}\}_{x=1}^{d}$ be an orthonormal basis for $\cH=\mathbb{C}^d$,
and let
\begin{equation}
    \ket{\phi} = \frac{1}{\sqrt{d}}\sum_{x=1}^{d^2}\ket{x}\otimes\ket{x}
\end{equation}
be a maximally entangled state on $\cH\otimes\cH$.
\begin{definition}
    The process fidelity of $\cE$, with respect to the identity, is given by
\begin{equation}
    F(\cE) = \bra{\phi}(\cI\otimes\cE)\big(\ket{\phi}\bra{\phi}\big)\ket{\phi},
\end{equation}
where $\cI$ denotes the identity process.
\end{definition}

The process fidelity turns out to be independent
of the choice of orthonormal basis $\{\ket{x}\}_x$.
The density operator $(\cI\otimes\cE)\big(\ket{\phi}\bra{\phi}\big)$ appearing in the above equation
is called the Choi state for the process $\cE$,
and is denoted $\chi_{\cE}$, or simply $\chi$ when
the process is understood.
While the stated definitions of the process and $0$-fidelities
are with respect to the identity,
we could also consider a target unitary $U$,
in which case the fidelities with respect to $U$ are defined as
$F(\cE,U)=F(\cU^{\dagger}\circ \cE)$ and $F_0(\cE,U)=F_0(\cU^{\dagger}\circ\cE)$.
The bounds given by our main result apply to any target unitary,
but without loss of generality we can take the fidelities
to be with respect to the identity.

In order to prove our main result,
we will need the following two facts.
The first is the expression for the output of
a process on a state in terms of its Choi matrix:
\begin{equation}\label{eq: E on rho}
    \cE(\rho) = d \Tr_1\big(\chi(\rho^{\intercal}\otimes I)\big),
\end{equation}
where the trace is over the first subsystem,
and $\intercal$ denotes transposition.
Next, for any operator $A$, the maximally entangled state satisfies
\begin{equation}\label{eq:A on phi}
    (I\otimes A)\ket{\phi} = (A^{\intercal}\otimes I)\ket{\phi}.
\end{equation}
These facts can be proven using the definitions.

\section{Process fidelity bounds}

Our main result is an upper and lower bound on the process fidelity 
as a linear function of the $0$-fidelity.

\begin{theorem}\label{theorem 1}
Let $\cE$ be an $n$-qubit process with $0$-fidelity $F_0$ and process fidelity $F$. Then
\begin{equation}
    1-\frac{3}{2}(1-F_0)\le F \le F_0.
\end{equation}
\end{theorem}

We will prove Th.~\ref{theorem 1} by a series of lemmas.
Define the operator $\Gamma$ by
\begin{equation}
    \Gamma = \sum_k \big(\ket{\psi_k}\bra{\psi_k}\big)^{\intercal}\otimes \ket{\psi_k}\bra{\psi_k}.
\end{equation}
By Eqs.~\eqref{eq: 0-fidelity} and~\eqref{eq: E on rho},
the $0$-fidelity can be rewritten as
\begin{equation}\label{eq: 0-fid versus Choi}
    F_0 = \frac{1}{d}\Tr(\chi\Gamma).
\end{equation}
Our first lemma gives a sufficient condition for Th.~\ref{theorem 1} to hold.
\begin{lemma}
    Let $A=\ket{\phi}\bra{\phi}-\frac{3}{2d}\Gamma+\frac{1}{2}I$ and $B=\frac{1}{d}\Gamma-\ket{\phi}\bra{\phi}$.
    Then the left and right inequalities in Th.~\ref{theorem 1} are implied by $A\ge0$ and $B\ge0$.
\end{lemma}
\begin{proof}
If $A\ge0$, then $\Tr(\chi A)\ge0$, since $\chi\ge0$.
We compute
\begin{align}
    \Tr(\chi A) &= \bra{\phi}\chi\ket{\phi}-\frac{3}{2d}\Tr(\chi \Gamma) + \frac{1}{2}\Tr(\chi)\notag\\
    &= F - \frac{3}{2} F_0 + \frac{1}{2}\notag\\
    &= F - \big(1 - \frac{3}{2}(1-F_0)\big).
\end{align}
Therefore, $A\ge0$ implies $F\ge1-\frac{3}{2}(1-F_0)$.
Similarly,
\begin{align}
    \Tr(\chi B) &= \frac{1}{d}\Tr(\chi\Gamma)-\bra{\phi}\chi\ket{\phi}\notag\\
    &= F_0-F,
\end{align}
and therefore, $B\ge0$ implies $F\le F_0$.
\end{proof}
The next lemma concerns $\mathrm{spec}(\Gamma)$, the spectrum of $\Gamma$.
\begin{lemma}
The largest eigenvalue of $\Gamma$ is $d$,
and the next largest eigenvalue is $d/3$.
These eigenvalues occur with multiplicities of $1$ and $3n$, respectively.
\end{lemma}
\begin{proof}
By defining
\begin{equation}
    \gamma = \sum_k\ket{\psi_k^*}\bra{k}\otimes\ket{\psi_k}\bra{k},
\end{equation}
where $^*$ denotes complex conjugation,
we see that $\Gamma=\gamma\gamma^{\dagger}$,
and therefore $\mathrm{spec}(\Gamma)=\mathrm{spec}(\gamma\gamma^{\dagger})=\mathrm{spec}(\gamma^{\dagger}\gamma)$.
Since
\begin{equation}
    \gamma^{\dagger}\gamma = \sum_{kk'}|\braket{\psi_k}{\psi_{k'}}|^2\ket{k}\bra{k'}\otimes\ket{k}\bra{k'},
\end{equation}
the non-zero eigenvalues of $\Gamma$ are equal
to those of the $d^2$-by-$d^2$ matrix $M$ with entries
\begin{align}
    M_{kk'} &=
     |\braket{\psi_k}{\psi_{k'}}|^2\notag\\
    &= \sum_{k_1k_1'=1}^4\cdots\sum_{k_nk_n'=1}^4 |\braket{\psi_{k_1}}{\psi_{k_1'}}|^2\cdots |\braket{\psi_{k_n}}{\psi_{k_n'}}|^2\notag\\
    &= \prod_{i=1}^n \bigg(\sum_{k_ik_i'=1}^4|\braket{\psi_{k_i}}{\psi_{k_i'}}|^2\bigg).
\end{align}
The  last line implies $M={M^{(1)}}^{\otimes n}$,
where $M^{(1)}$ is the matrix of overlaps of single-qubit SIC-POVM states.
\begin{equation}
    M^{(1)}=\begin{pmatrix}
    1 & 1/3 & 1/3 & 1/3\\
    1/3 & 1 & 1/3 & 1/3\\
    1/3 & 1/3 & 1 & 1/3\\
    1/3 & 1/3 & 1/3 & 1
    \end{pmatrix}
\end{equation}
The eigenvalues of $M^{(1)}$ are $2$ and $2/3$,
occurring with multiplicities $1$ and $3$, respectively.
It follows that the largest eigenvalue of $\Gamma$ is $2^n=d$,
and the next largest eigenvalue is $2^{n-1}(2/3)=d/3$, with multiplicity $3n$.
\end{proof}

\begin{lemma}
$\ket{\phi}$ is the eigenvector of $\Gamma$ with eigenvalue $d$.
\end{lemma}
\begin{proof}
Indeed,
\begin{align}
    \Gamma\ket{\phi} &= \sum_k\big(\ket{\psi_k}\bra{\psi_k}\big)^{\intercal}\otimes\big(\ket{\psi_k}\bra{\psi_k}\big)\ket{\phi}\notag\\
    &= \sum_k\big(\ket{\psi_k}\bra{\psi_k}\big)^{\intercal}\otimes I\ket{\phi}\notag\\
    &= \sum_{k_1\dots k_n}\bigotimes_{i=1}^n\big(\ket{\psi_{k_i}}\bra{\psi_{k_i}}\big)^{\intercal}\otimes I \ket{\phi}\notag\\
    &= \bigotimes_{i=1}^n\bigg( \sum_{k_i}\big(\ket{\psi_{k_i}}\bra{\psi_{k_i}}\big)^{\intercal}\bigg)\otimes I\ket{\phi}\notag\\
    &= \bigotimes_{i=1}^n\bigg(\sum_{k_i}\ket{\psi_{k_i}}\bra{\psi_{k_i}}\bigg)^{\intercal}\otimes I\ket{\phi}\notag\\
    &= 2^n\ket{\phi},
\end{align}
where in the 2nd line we used Eq.~\eqref{eq:A on phi}
and in the last line we used the fact that 
single-qubit SIC-POVMs satisfy $\sum_{k_i}\ket{\psi_{k_i}}\bra{\psi_{k_i}}=2I$.
\end{proof}

\begin{proof}[Proof of Theorem~\ref{theorem 1}]

We first prove that the operator $B$
in Lemma~1 satisfies $B\ge0$. By definition $\Gamma\ge0$, and
by Lemma~3, $\Gamma\ket{\phi}=d\ket{\phi}$.
It follows that $B=\frac{1}{d}\Gamma-\ket{\phi}\bra{\phi}\ge0$.

We next show that $A\ge0$.
Decompose $\Gamma=d\ket{\phi}\bra{\phi} + \Gamma_{\perp}$,
where $\Gamma_{\perp}$ is the component of $\Gamma$ that
is orthogonal to $\ket{\phi}\bra{\phi}$.
Similarly, decompose the identity as $I=\ket{\phi}\bra{\phi}+I_{\perp}$.
The operator $A$ is then given by
\begin{align}
    A&=\ket{\phi}\bra{\phi}-\frac{3}{2d}\Gamma+\frac{1}{2}I\notag\\
    &= \ket{\phi}\bra{\phi} - \frac{3}{2}\ket{\phi}\bra{\phi} -\frac{3}{2d}\Gamma_{\perp} + \frac{1}{2}\ket{\phi}\bra{\phi} + \frac{1}{2}I_{\perp}\notag\\
    &= -\frac{3}{2d}\Gamma_{\perp} + \frac{1}{2}I_{\perp}.
\end{align}
By Lemma~2, the largest eigenvalue of $\Gamma_{\perp}$ is $d/3$. 
We therefore have
\begin{equation}
    A=\frac{1}{2}\big(I_{\perp} - \frac{3}{d}\Gamma_{\perp}\big)\ge0,
\end{equation}
which completes the proof.

\end{proof}

\section{Tight bounds via semidefinite programming}

We next provide evidence that the lower bound given
by Th.~1 is tight.
The minimum and maximum of $F(\cE)$
over the set of processes having a given $0$-fidelity
can be found via semidefinite programming~\cite{Boyd2004}.
The CP and TP constraints on $\cE$ translate to 
the constraints that $\chi\ge 0$ and $\Tr_2(\chi)=I/d$, respectively~\cite{Audenaert2002}.
Using Eq.~\eqref{eq: 0-fid versus Choi} for
the constraint on $F_0(\cE)$,
the minimum is found by solving
\begin{align}
    \mathrm{Minimize}:\quad&\Tr(\chi\ket{\phi}\bra{\phi}),\notag\\
    \mathrm{Subject}\,\mathrm{to} :\quad&\Tr(\chi\Gamma)=dF_0,\notag\\
    &\Tr_2(\chi)=I/d,\notag\\
    &\,\chi\ge 0.
\end{align}
We solve this SDP using the cvx package~\cite{cvx},
and find that the solution is $1-\frac{3}{2}(1-F_0)$, independent of $n$ for $n\le4$.
This leads us to conjecture that the lower bound in Th.~1 is tight for all $n$.

\begin{figure}[ht]
\centering
\includegraphics[scale=0.38]{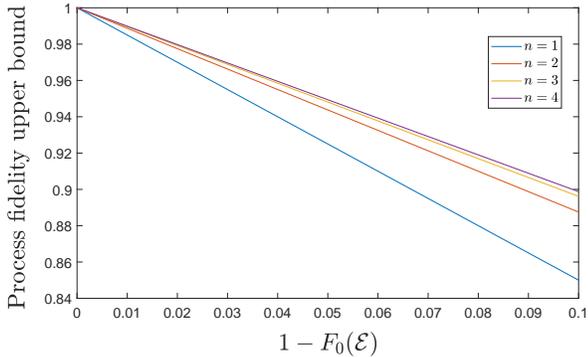}
\centering
 \caption{SDP solutions to~\eqref{eq: SDP upper bound} giving tight upper bounds
 on $F(\cE)$ as a function of $1-F_0(\cE)$ and
 qubit number $n$.}
 \label{fig: upper bounds}
\end{figure}

Similarly, a tight upper bound on $F(\cE)$
is given by solving
\begin{align}\label{eq: SDP upper bound}
    \mathrm{Maximize}:\quad&\Tr(\chi\ket{\phi}\bra{\phi}),\notag\\
    \mathrm{Subject}\,\mathrm{to} :\quad&\Tr(\chi\Gamma)=dF_0,\notag\\
    &\Tr_2(\chi)=I/d,\notag\\
    &\,\chi\ge 0.
\end{align}
In this case, the solutions are strictly less than $F_0(\cE)$,
implying that the upper bound of $F(\cE)\le F_0(\cE)$ in Th.~1
is not tight, but rather becomes an increasingly better approximation as $n$ increases.
These results are shown in Fig~\ref{fig: upper bounds}.
For $n=1$, $\{\ket{\psi_k}\}_k$ is a 2-design and therefore $F_0(\cE)=F_{\mathrm{avg}}(\cE)$, the average fidelity of $\cE$,
given by~\cite{Nielsen2002}
\begin{equation}
    F_{\mathrm{avg}}(\cE)=\frac{d F(\cE)+1}{d+1}.
\end{equation}
Therefore, for $n=1$, the process fidelity is uniquely determined by the $0$-fidelity,
and is equal to the lower bound.
For $n>1$, the tight upper bound to the process fidelity converges to $F_0(\cE)$ as $n$ increases.

\section{discussion}

We have proven a lower bound on the process 
fidelity in terms of the $0$-fidelity,
namely $1-\frac{3}{2}(1-F_0)\le F$,
and by numerically solving an SDP for small $n$ conjecture that this bound is tight.
The authors of~\cite{greenaway2021} numerically
calculated $F_0(\cE)$ and $F(\cE)$ for an ensemble
of random processes and concluded that $F_0(\cE)$ converges
to $F(\cE)$ with increasing $n$. While true for the ensemble of
processes used in that simulation,
if our conjecture holds, then
the worst case process fidelity does not converge to the
$0$-fidelity with increasing $n$. A description of the worst-case process is an open problem.

We have also computed the best-case process fidelity
as the solution to an SDP.
This solution does converge to the $0$-fidelity with
increasing $n$, qualitatively matching the results from~\cite{greenaway2021},
and improves upon the protocol in that work for estimating the process fidelity.
Specifically, given an estimate of the $0$-fidelity,
the process fidelity may be estimated as the upper bound
given by the SDP solution, rather than by the $0$-fidelity itself.

\acknowledgements{The author thanks Sean Greenaway for feedback on this manuscript while in preparation.}

\nocite{*}

\bibliography{library}

\end{document}